\documentclass{article}
\usepackage{amsthm,framed,amssymb}
\usepackage[sumlimits]{amsmath}
\usepackage{amsfonts,enumerate}
\usepackage{color,url}
\usepackage[margin=2cm]{geometry}
\usepackage{graphicx}  

\DeclareMathOperator{\diff}{d}
\DeclareMathOperator{\dx}{\diff{x}}

\DeclareMathOperator{\Diff}{Diff}
\DeclareMathOperator{\Id}{Id}

\newtheorem{theorem}{Theorem}

\newtheorem{corollary}[theorem]{Corollary}

\newtheorem{proposition}[theorem]{Proposition}

\def\MM#1{\boldsymbol{#1}}
\newcommand{\pp}[2]{\frac{\partial #1}{\partial #2}} 
\newcommand{\dede}[2]{\frac{\delta #1}{\delta #2}}
\newcommand{\dd}[2]{\frac{\diff #1}{\diff #2}}

\newcommand{\rem}[1]{}
\newcommand{\revised}[1]{{\color{blue} #1}}
\def\contract{\makebox[1.2em][c]{\mbox{\rule{.6em}
{.01truein}\rule{.01truein}{.6em}}}}
\numberwithin{equation}{section}
\numberwithin{figure}{section}

\makeatother

\begin{document}

\title{Particle relabelling symmetries and Noether's theorem for vertical slice models}
\author{C. J. Cotter$^{1}$ and M.J.P. Cullen$^{2}$}

\addtocounter{footnote}{1}
\footnotetext{Department of Mathematics, Imperial College London. London SW7 2AZ, UK. 
\texttt{colin.cotter@imperial.ac.uk}
\addtocounter{footnote}{1} }
\footnotetext{Met Office, Fitzroy Road, Exeter EX1 3PB.
\addtocounter{footnote}{1}}
\maketitle

\begin{abstract}
  We consider the variational formulation for vertical slice models
  introduced in Cotter and Holm (Proc Roy Soc, 2013). These models
  have a Kelvin circulation theorem that holds on all
  materially-transported closed loops, not just those loops on
  isosurfaces of potential temperature. Potential vorticity
  conservation can be derived directly from this circulation theorem.
  In this paper, we show that this property is due to these models
  having a relabelling symmetry for every single diffeomorphism of the
  vertical slice that preserves the density, not just those
  diffeomorphisms that preserve the potential temperature. This is
  developed using the methodology of Cotter and Holm (Foundations of
  Computational Mathematics, 2012).
\end{abstract}

\tableofcontents

\section{Introduction}

Vertical slice models are models of 3D fluids that assume that all
fields are independent of $y$, except for the (potential) temperature,
which is assumed to take the form $\theta(x,y,z,t)=\theta_S(x,z,t) +
sy$, where $s$ is a time-independent coefficient, this allows for a
model of North-South temperature gradient $s$. These models provide a
simplification of the equations that allows to study geophysical
phenomena such as frontogenesis in idealised geometries. They also
provide a useful testbed for numerical algorithms for atmospheric
dynamical cores, since they only require computation with 2D data, and
so can be run quickly on a desktop computer.

A hierarchy of vertical slice models have been used to study the
formation and evolution of fronts. The vertical slice non-hydrostatic
incompressible Boussinesq equations, the hydrostatic Boussinesq
equations and the corresponding semi-geostrophic equations all exhibit
frontogenesis, whilst representing solutions of the full 3D
equations. As summarised in \cite{Cu2006}, the semi-geostrophic
equations have an optimal transport interpretation. The optimal
transport formulation proves that geostrophic and hydrostatic balance
can be achieved while respecting Lagrangian conservation properties.
The optimal transport formulation has also been used to develop
numerical algorithms for the equations in Lagrangian form, that form
fronts even at low resolution. These solutions provide a useful
comparison for standard Eulerian numerical algorithms for the
non-hydrostatic incompressible Boussinesq equations, by considering
their solutions in the semi-geostrophic limit, as described in
\cite{visram2014framework} who found that the limiting Eulerian
solutions deviated from the semi-geostrophic solution, suggesting that
fronts need some extra parameterisation in Eulerian models. The
limiting Eulerian numerical solutions satisfied geostrophic and
hydrostatic balance to the expected extent, but the Lagrangian
conservation properties were systematically violated. This shows that the
computed solution was fundamentally diffusive, which is not physically
realistic on this scale. Thus, there is a strong incentive to devise
computable formulations which retain Lagrangian conservation
properties, as stated in \cite{Cu2007}.

It would be even more useful to operational weather forecasting to be
able to make these comparisons with compressible models. \cite{Cu2008}
attempted to do this, but encountered the problem that it is not
possible to make a compressible vertical slice model whose solutions
represent solutions of the full 3D compressible equations. A vertical
slice model was proposed that can be obtained as a small modification
of a 3D compressible model, but since the model did not conserve
energy, it was not possible to make meaningful comparisons on a long
timescale. Then, \cite{cotter2013variational}, developed a variational
framework for deriving fluid models in the vertical slice framework
that conserve energy and potential vorticity. As well as recovering
the incompressible Euler-Boussinesq equations,
\cite{cotter2013variational} provided a number of new models,
including an $\alpha$-regularised vertical slice model, and a
compressible vertical slice model, which fixed the lack of
conservation in the compressible model of \cite{Cu2008}.

In \cite{cotter2013variational}, it was noted that the slice
Euler-Poincar\'e equations with advected density $D$ and potential
temperature $\theta_S$ have a Kelvin circulation theorem
\begin{align}
\frac{d}{dt}\oint_{c(u_S)}\hspace{-1mm}
\left(s\left(\frac1D\dede{l}{u_S}\right)  - \left(\frac1D\dede{l}{u_T}\right)\,\nabla\theta_S  \right)\cdot \dx
=
0
\,,
\label{EPSD-circthm}
\end{align}
where $c(u_S)$ is a loop advected by the in-slice velocity $u_S$,
$u_T$ is the transverse velocity (the component orthogonal to the
slice), \revised{the potential temperature is written as
\begin{equation}
  \theta(x,y,z,t) = \theta_S(x,z,t) + sy,
\end{equation}
where $s$ is the constant temperature gradient in the $y$-direction,}
and $l[u_S,u_T,\theta_S,D]$ is the reduced Lagrangian. For the
case of the incompressible Euler-Boussinesq slice model, this becomes
\begin{equation}
  \dd{}{t}\oint_{C(t)}\left(su_S
  - (u_T + fx)\nabla\theta_S\right)\cdot \diff x = 0.
\end{equation}
This is curious because the circulation theorem for EP equations on
the diffeomorphism group (instead of the semidirect product used in
the slice models) has a baroclinic term that only vanishes which the
curve is on a $\theta$-isosurface. In the slice model case, this term
is replaced by an additional term in the circulation, leading to a
conserved circulation on \emph{any loop}, $\theta$-isosurface or not.
Since the contour used to defined the circulation theorem is a
material contour, it cannot cross fronts in the semi-geostrophic
limits. In the semi-geostrophic solutions, fronts are singularities in
the optimally transporting map where the boundary of the domain has
been mapped into the interior. Thus solutions which are discontinuous
in physical space are still compatible with the circulation theorem,
as long as this restriction is made.  These formulations will preserve
Lagrangian conservation properties. However, computing the solutions
remains challenging. Care with use potential vorticity conservation
near the semi-geostrophic limit for standard models is therefore also
required.

In this paper, we show that this Kelvin circulation theorem occurs
because of relabelling symmetries that occur in the variational
description of \cite{cotter2013variational}, by using the variational
tools for applying Noether's Theorem to the Euler-Poincar\'e theory
discussed in \cite{cotter2013noether}. The rest of this paper is structured as follows. In Section
\ref{sec:variational}, we review the variational formulation of
vertical slice models. In Section 
\ref{sec:relabelling}, we show how the conserved potential vorticity
arises from relabelling symmetries that are specific to the slice
geometry. Finally in Section
\ref{sec:summary} we provide a summary and outlook.

\section{Variational formulations of slice models}
\label{sec:variational}
In this section, we briefly review the variational formulation of
slice models. In a vertical slice model, we assume that the velocity
field is independent of $y$, but still has a component in the
$y$-direction. This means that the Lagrangian flow map can be
written
\begin{equation}
  \Phi(X,Y,Z,t) = (x(X,Z,t), y(X,Z,t) + Y, z(X,Z,t)),
\end{equation}
where $X,Y,Z$ are Lagrangian labels, $(x,y,z)$ are particle locations
and $t$ is time. Equivalently, the Lagrangian flow map can be
represented by a diffeomorphism\footnote{The assumption that the
  invertible map $\Phi$ is smooth is necessary to compute the
  potential vorticity equation. However, the optimal transport
  formulation relaxes the smoothness requirements, enabling weaker
  solutions such as the semigeostrophic frontal solutions which are
  weak Lagrangian solutions. The Lagrangian formulation is likely to
  have some sort of global existence, based on conservation
  properties, but the Eulerian formulation may only work in smooth
  cases, which may not be generic for weather phenomena such as
  fronts.} of the vertical slice $x-z$ plane \revised{combined with an}
$(x,z)$-dependent displacement in the $y$-direction, which we write as
an element $(\phi, f)$ of the semidirect product
$\Diff(\Omega)\circledS \mathcal{F}(\Omega)$, where $\circledS$
denotes the semidirect product, and $\mathcal{F}(\Omega)$ denotes an
appropriate space of smooth functions on $\Omega$ that specify the
displacement of Lagrangian particles in the $y$-direction at each
point in $\Omega$. \revised{Then, if we write $\MM{X}=(x,z)$, the transformation
takes the form
\begin{equation}
(\MM{x},y) =  \Phi(\MM{X},Y) = (\phi(\MM{X}), Y + f(\MM{X})).
\end{equation}}

In this representation, composition of two slice flow maps $(\phi_1,f_1)$
and $(\phi_2,f_2)$ is obtained from the semi-direct product formula
\cite{HoMaRa1998},
\begin{equation}
(\phi_1,f_1)\cdot(\phi_2,f_2) = (\phi_1\circ\phi_2, f_1\circ \phi_2 + f_2).
\label{SDgroupaction}
\end{equation}
We see that the vertical slice diffeomorphisms $\phi_1$ and $\phi_2$
compose in the normal way, whilst the combined vertical deflection is
obtained by moving $f_1$ with $\phi_2$ before adding $f_2$.

We represent Eulerian velocity fields by splitting into two components
$(u_S,u_T)$ where $u_S$ is the ``slice'' component in the $x$-$z$
plane, and $u_T$ is the ``transverse'' component in the $y$ direction.
$(u_S,u_T)$ is considered as an element in the semidirect product Lie
algebra $\mathfrak{X}(\Omega)\circledS \mathcal{F}(\Omega)$ where
$\mathfrak{X}(\Omega)$ denotes the vector fields on $\Omega$,
representing the two components of the velocity
$u_S\in\mathfrak{X}(\Omega)$ and $u_T\in \mathcal{F}(\Omega)$. This
Lie algebra has a Lie bracket, which we shall make use of below, given
by
\begin{equation}
[(u_S,u_T),(w_S,w_T)] = \left([u_S,w_S] ,\,u_S\cdot\nabla w_T -
w_S\cdot\nabla u_T \right),
\end{equation}
where $[u_S,w_S]=u_S\cdot\nabla w_S - w_S\cdot\nabla u_S$ is the Lie
bracket for the time-dependent vector fields
$(u_S,w_S)\in\mathfrak{X}(\Omega)$, and $\nabla$ denotes the gradient
in the $x$-$z$ plane.

A time-dependent Lagrangian flow map must satisfy the equation
\begin{equation}
  \label{eq:flow map}
  \pp{}{t}(\phi,f) = (u_S,u_T)(\phi,f) = (u_S\circ\phi, u_T \circ \phi),
\end{equation}
for some slice vector field $(u_S,u_T)\in
\mathfrak{X}(\Omega)\circledS \mathcal{F}(\Omega)$. Similarly, if
within Hamilton's principle we consider a one-parameter family of
perturbations $(\phi_\epsilon,f_\epsilon)$ (parameterised by $\epsilon$)
to $(\phi,f)$, then
\begin{equation}
  \delta(\phi,f) = \lim_{\epsilon\to 0}\frac{(\phi_\epsilon,f_\epsilon)-(\phi,f)}{\epsilon} = (w_S,w_T)(\phi,f) = (w_S\circ\phi, w_T\circ \phi),
\end{equation}
for some time-dependent slice vector field $(u_S,u_T)$ that generates
the infinitesimal perturbations at $\epsilon=0$. Taking $\epsilon$ and
time-derivatives of these two expressions and comparing leads to
\begin{equation}
  \delta(u_S,u_T) = \pp{}{t}(w_S,w_T) + [(u_S,u_T),(w_S,w_T)],
\end{equation}
which gives us a formula telling us how perturbations in $(\phi,f)$ lead
to perturbations in $(u_S,u_T)$ that we can use in Hamilton's principle.

In order to build geophysical models, we also need to consider
advected tracers and densities. In this modelling framework, we assume
that densities are independent of $y$, so that the continuity equation
becomes
\begin{equation}
  \pp{}{t}D + \nabla\cdot(u_SD) = 0,
\end{equation}
which we write in geometric notation as
\begin{equation}
  \pp{}{t}D\diff S + \mathcal{L}_{u_S}D\diff S = 0,
\end{equation}
where $\diff S$ is the area form (so that $D\diff S\in
\Lambda^2(\Omega)$, the space of 2-forms on $\Omega$) and
$\mathcal{L}_{u_S}$ is the Lie derivative with respect to $u_S$,
given via Cartan's Magic Formula as 
\begin{equation}
  \mathcal{L}_{u_S}D\diff S = \diff (u_S\contract D \diff S),
\end{equation}
recalling that $\diff(D \diff S)=0$. Under an infinitesimal
perturbation to $(\phi,f)$ the density changes according to
\begin{equation}
  \label{eq:D}
  \delta D\diff S + \mathcal{L}_{w_S}D\diff S \revised{=0}.
\end{equation}

We assume that tracers can be written as sum of a $y$-independent
component $\theta_S$ plus a time-independent component $sy$ with
linear dependence on $y$ (geophysically this allows for North-South
temperature gradients that are necessary for the baroclinic
instability that leads to frontogenesis). \revised{This means 
  that the flow map $(\phi, f)$ transports initial conditions $(\theta_{S,0},s_0)$ according
  to
  \begin{equation}
    \theta_S = \phi_*(\theta_{S_0} - s_0f), \quad s = s_0.
    \end{equation}
  Time-differentiation then leads to the transport equation}
\begin{equation}
  \pp{}{t}\theta_S + u_S\cdot\nabla \theta_S \revised{+sw_T} = 0, \quad \pp{s}{t}=0,
\end{equation}
which we write in geometric notation as
\begin{equation}
  \pp{}{t}(\theta_S,s) + \mathcal{L}_{(u_S,u_T)}(\theta_S,s) = 0, 
\end{equation}
where
\begin{equation}
  \mathcal{L}_{(u_S,u_T)}(\theta_S,\,{s} )
  =\left(u_S\contract \diff \theta_S + s u_T, 0\right),
\end{equation}
and where we have considered $(\theta_S,s)\in
\mathcal{F}(\Omega)\times \mathbb{R})$.  Similarly, infinitesimal
perturbations to $(\phi,f)$ lead to infinitesimal perturbations to
$(\theta_S,s)$ given by
\begin{equation}
  \delta(\theta_S,s) + \mathcal{L}_{(w_S,w_T)}(\theta_S,s) = 0.
\end{equation}

Given a Lagrangian $l[(u_S,u_T),D,(\theta_S,s)]$, we define
\begin{eqnarray}
  \nonumber
&  \delta l[(u_S,u_T),D,(\theta_S,s); (\delta u_S,\delta u_T),\delta D,
    (\delta \theta_S, \delta s)] \\
  =& \lim_{\epsilon\to 0}\frac{1}{\epsilon}
    \left(l[(u_S,u_T) + \epsilon(\delta u_S,\delta u_T),D + \epsilon \delta D,
      (\theta_S,s) + \epsilon (\delta \theta_S,\delta s)]
    -l[(u_S,u_T),D,(\theta_S,s)]\right)
 \nonumber  , \\
  &= \int_\Omega \dede{l}{u_S}\cdot \delta u_S + \dede{l}{u_T}\delta u_T
  + \dede{l}{D}\delta D + \dede{l}{\theta_S}\delta\theta_S + \dede{l}{s}\delta s
  \diff S, 
\end{eqnarray}
for all $(\delta u_S,\delta u_T) \in \mathfrak{X}(\Omega)\circledS
\mathcal{F}(\Omega),\delta D \in \Lambda^2(\Omega), (\delta \theta_S,
\delta s) \in \mathcal{F}(\Omega)\times \mathbb{R}$. Geometrically, we
interpret $\dede{l}{u_S}\cdot \diff x \otimes \diff S$ as being a
1-form-density in $\Lambda^1(\Omega)\otimes \Lambda^2(\Omega)$, the
dual space of vector fields $\mathfrak{X}(\Omega)$. $\dede{l}{v_T}\diff S$
and $\dede{l}{\theta_S}\diff S$ are interpreted as being densities in
$\Lambda^2(\Omega)$, whilst $\dede{l}{D}$ is interpreted as being a
function in $\mathcal{F}(\Omega)$, and $\dede{l}{s}\in \mathbb{R}$.

Proceeding with Hamilton's principle $\delta S = 0$, where
\begin{equation}
  \delta S = \int_{t_0}^{t_1} l[(u_S,u_T),D,(\theta_S,s)] \diff t,
\end{equation}
considering variations 
\begin{align}
  \delta(u_S,u_T) &= \pp{}{t}(w_S,w_T) + [(u_S,u_T),(w_S,w_T)], \\
  \delta D\diff S &=  -\mathcal{L}_{w_S}D\diff S, \\
  \delta(\theta_S,s) &= -\mathcal{L}_{(w_S,w_T)}(\theta_S,s), 
\end{align}
for perturbation-generating velocity fields $(w_S,w_T)$ that
vanish at $t=t_0$ and $t=t_1$, we obtain
\begin{align}
\nonumber  0  &= \delta S \\ 
\nonumber  &= \int_{t_0}^{t_1} \int_\Omega \dede{l}{u_S}\cdot \delta u_S + \dede{l}{u_T}\delta u_T
  + \dede{l}{D}\delta D + \dede{l}{\theta_S}\delta\theta_S + \dede{l}{s}\delta s
  \diff S \diff t, \\
\nonumber  &= \int_{t_0}^{t_1} \int_\Omega \dede{l}{u_S}\cdot \left(\pp{}{t}w_S + [u_S,w_S]\right)\diff S + \dede{l}{u_T}\left(\pp{}{t}w_T + \mathcal{L}_{u_S}w_T 
  - \mathcal{L}_{w_S}u_T\right)\diff S \\
\nonumber& \qquad  - \dede{l}{D}\mathcal{L}_{w_S}(D\diff S) -
  \mathcal{L}_{(w_S,w_T)}\theta_S
  \dede{l}{\theta_S}\diff S \diff t, \\
 \nonumber &= -\int_{t_0}^{t_1} \int_\Omega w_S\cdot \left(
  \left(\pp{}{t} + \mathcal{L}_{u_S}\right)\dede{l}{u_S}\cdot\diff x \otimes
  \diff S + \dede{l}{u_T}\diff u_T\otimes \diff S + \dede{l}{\theta_S}\diff \theta_S \otimes \diff S
  - D\diff \dede{l}{D} \otimes \diff S\right) \nonumber \\
 &\qquad  + \int_\Omega w_T\left(\left(\pp{}{t} + \mathcal{L}_{u_S}\right)\dede{l}{u_T}\diff S 
  + \dede{l}{\theta_S}s\diff S\right)\diff t,
\end{align}
where we have integrated by parts in time and space, and have used the identity
\begin{equation}
  \int_\Omega m\cdot [u,v]\otimes \diff S = -\int_\Omega v\cdot\mathcal{L}_{u_S}
  m\cdot x\otimes S,
\end{equation}
for all $u,v \in \mathfrak{X}(\Omega)$, $m\cdot\diff x \otimes \diff S
\in \Lambda^1(\Omega)\otimes \Lambda^2(\Omega)$, and where the
Lie derivative of a one-form density is defined as
\begin{equation}
  \mathcal{L}_{u_S} (m\cdot \diff x \otimes \diff S)
  = (\mathcal{L}_{u_S} m\cdot \diff x)\otimes \diff S + m\cdot \diff x
  \otimes \diff S.
\end{equation}

Since $(w_S,w_T)$ are arbitrary, save for the endpoint conditions, we
obtain (for sufficiently smooth solutions),
\begin{align}
  \label{eq:EPS}
  \left(\pp{}{t} + \mathcal{L}_{u_S}\right)\dede{l}{u_S}\cdot\diff x \otimes
  \diff S + \dede{l}{u_T}\diff u_T\otimes \diff S + \dede{l}{\theta_S}\diff \theta_S \otimes \diff S
  - D\diff \dede{l}{D} \otimes \diff S & = 0, \\
  \label{eq:EPT}
  \left(\pp{}{t} + \mathcal{L}_{u_S}\right)\dede{l}{u_T}\diff S 
  + \dede{l}{\theta_S}s\diff S & = 0.
\end{align}
It becomes useful to write
\begin{equation}
  \dede{l}{u_S}\cdot x\otimes \diff S = \frac{1}{D}\dede{l}{u_S}\cdot x
  \otimes D \diff S, 
\end{equation}
in which case
\begin{equation}
  \left(\pp{}{t} + \mathcal{L}_{u_S}\right)\frac{1}{D}\dede{l}{u_S}\cdot x\otimes D \diff S = 
  \left(\left(\pp{}{t} + \mathcal{L}_{u_S}\right)\frac{1}{D}\dede{l}{u_S}\cdot x\right) 
  \otimes D \diff S,
\end{equation}
after making use of the continuity equation for $D\diff S$. Similarly,
we have
\begin{equation}
  \left(\pp{}{t} + \mathcal{L}_{u_S}\right)\frac{1}{D}\dede{l}{u_T} D \diff S = 
  \left(\left(\pp{}{t} + \mathcal{L}_{u_S}\right)\frac{1}{D}\dede{l}{u_T}\right) D \diff S.
\end{equation}
Hence we obtain
\begin{align}
  \label{eq:dlduS/D}
  \left(\pp{}{t} + \mathcal{L}_{u_S}\right)\frac{1}{D}\dede{l}{u_S}\cdot\diff x + \frac{1}{D}\dede{l}{u_T}\diff u_T + \frac{1}{D}\dede{l}{\theta_S}\diff \theta_S
  - \diff \dede{l}{D} & = 0, \\
  \label{eq:dlduT/D}
  \left(\pp{}{t} + \mathcal{L}_{u_S}\right)\frac{1}{D}\dede{l}{u_T}
  + \frac{1}{D}\dede{l}{\theta_S}s & = 0.
\end{align}

Translating back into vector calculus notation, we get
\begin{align}
  \left(\pp{}{t} + u_S \cdot \nabla + (\nabla u_S)^T\cdot \right)\frac{1}{D}\dede{l}{u_S}
  + \frac{1}{D}\dede{l}{u_T}\nabla u_T + \frac{1}{D}\dede{l}{\theta_S}\nabla \theta_S
  - \nabla \dede{l}{D} & = 0, \\
  \left(\pp{}{t} + u_S\cdot\nabla\right)\frac{1}{D}\dede{l}{u_T}
  + \frac{1}{D}\dede{l}{\theta_S}s & = 0.
\end{align}
The Lagrangian for the incompressible Euler-Boussinesq slice equations is
\begin{equation}
  l = \int_\Omega \frac{D}{2}\left(|u_S|^2 + u_T^2\right) + Dfu_T x
  + \frac{g}{\theta_0}D\left(z-\frac{H}{2}\right)\theta_S + p(1-D)\diff S,
\end{equation}
where $f$ is the Coriolis parameter, $g$ is the acceleration due to
gravity, $\theta_0$ is a reference potential temperature, $H$ is the height
of the vertical slice, and $p$ is a Lagrange multiplier introduced to
enforce that the density stays constant. In this case we have
\begin{align}
\begin{split}
\frac1D\dede{l}{u_S}  =  u_S
\,, \qquad
\frac1D\dede{l}{u_T}  =  u_T + fx
\,, \\
\dede{l}{D}  =  \frac{1}{2}\left(|u_S|^2 + u_T^2\right)
+ fu_Tx
- p + \frac{g}{\theta_0}\theta_S\left(z-\frac{H}{2}\right)
, \\
\frac1D\dede{l}{\theta_S}  
=  \frac{g}{\theta_0}\left(z-\frac{H}{2}\right).
\end{split}
\end{align}
Substituting these formula and rearranging leads us to the
Euler-Boussinesq vertical slice equations
\begin{align}
\begin{split}
\partial_t u_S + u_S\cdot\nabla u_S 
-fu_T{\hat{x}}
&= -\nabla p 
 + \frac{g}{\theta_0}\theta_S\hat{z}, \\
\partial_t u_T + u_S\cdot\nabla u_T 
+ fu_S\cdot\hat{{x}}
& =   -\frac{g}{\theta_0}\left(z-\frac{H}{2}\right){s}, \\
\nabla\cdot u_S & =  0, \\
\partial_t \theta_S + u_S\cdot\nabla\theta_S + u_T{s} & =  0,
\end{split}
\label{Eady-EPSDeqns1}
\end{align}
where $\hat{{x}}$ and $\hat{z}$ are the unit normals in the $x$-
and $z$- directions, respectively. In addition to providing a
variational derivation of the incompressible Euler-Boussinesq slice
model, \cite{cotter2013variational} also provided Lagrangians that
lead to an alpha-regularised Euler-Boussinesq slice model, and a
compressible Euler slice model. Since these models have a variational
derivation, they all have a conserved energy; they also have a conserved
potential vorticity as we shall now discuss.

Returning to (\ref{eq:dlduS/D}-\ref{eq:dlduT/D}), 
\cite{cotter2013variational} made the following direct calculation
to derive Kelvin's circulation theorem and thus conservation
of potential vorticity. Using the fact that the exterior derivative $\diff$
commutes with $\partial/\partial t$ and $\mathcal{L}_{u_S}$, we deduce
that
\begin{equation}
  \left(\pp{}{t} + \mathcal{L}_{u_S}\right)\diff \theta_S = -s \diff u_T.
\end{equation}
Combining with \eqref{eq:dlduT/D}, we obtain that
\revised{
\begin{align}
\nonumber\left(\pp{}{t} + \mathcal{L}_{u_S}\right)\frac{1}{D}\dede{l}{u_T}\diff\theta_S &=
  \frac{1}{D}\dede{l}{u_T} \left(\pp{}{t} + \mathcal{L}_{u_S}\right)
  \diff \theta_S + \diff\theta_S\left(\pp{}{t} + \mathcal{L}_{u_S}\right)\frac{1}{D}\dede{l}{u_T}, \\
\nonumber  & = -s\left(\frac{1}{D}\dede{l}{u_T}\diff u_T + \frac{1}{D}\dede{l}{\theta_S} \diff \theta_S\right), \\
  & = \left(\pp{}{t} + \mathcal{L}_{u_S}\right)\frac{1}{D}\dede{l}{u_S}\cdot\diff x - \diff \dede{l}{D},
\end{align}
where we used \eqref{eq:dlduS/D} in the last equality.}
Hence, we obtain
\begin{equation}
  \label{eq:vorticity}
  \left(\pp{}{t} + \mathcal{L}_{u_S}\right)\left(s\frac{1}{D}\dede{l}{u_\revised{S}}\revised{\cdot \diff x}
  \revised{-} \frac{1}{D}\dede{l}{\theta_S}\diff \theta_S\right) = \diff \dede{l}{D}.
\end{equation}
Integrating this around a closed curve $C(t)$ that is moving with
velocity $u_S$, we obtain a Kelvin circulation theorem
\begin{equation}
  \label{eq:kelvin}
  \dd{}{t}\oint_{C(t)}\left(s\frac{1}{D}\dede{l}{u_\revised{S}}\revised{\cdot \diff x}
  \revised{-} \frac{1}{D}\dede{l}{\theta_S}\diff \theta_S\right) = 0.
\end{equation}
In the case of the incompressible Euler Boussinesq slice model, the
circulation theorem reads
\begin{equation}
  \dd{}{t}\oint_{C(t)}\left(su_S
  - (u_T + fx)\nabla\theta_S\right)\cdot \diff x = 0.
\end{equation}
\revised{Applying $\diff$ to Equation \eqref{eq:vorticity}, we obtain}
\begin{equation}
  \left(\pp{}{t} + \mathcal{L}_{u_S} \right)\diff \left(s\frac{1}{D}\dede{l}{u_T}
  \revised{-} \frac{1}{D}\dede{l}{\theta_S}\diff \theta_S\right) = 0.
\end{equation}
Finally combining with the continuity equation we obtain conservation
of potential vorticity,
\revised{
\begin{align}
  \left(\partial_t + \mathcal{L}_{u_S}\right)q = 0,
\end{align}
\begin{equation}
    q = \frac{1}{D} \diff \left(s\frac{1}{D}\dede{l}{u_S} -
    \frac{1}{D}\dede{l}{u_T}\diff\theta_S\right).
  \end{equation}
In the case of the incompressible Euler-Boussinesq slice model this becomes
\begin{equation}
    q = s\nabla^\perp\cdot \left({u}_S -
    (u_T + fx)\nabla\theta_S\right).
\end{equation}
}
The circulation theorem is curious, because usually in the presence of
advected temperatures, we obtain a baroclinic source term, so that
circulation is only preserved on isosurfaces of $\theta_S$. In the
slice model case, this baroclinic term can be replaced by the Lie
derivative of an additional quantity, so that we get a conservation
law for any circulation loop. The new contribution of this paper
is to show that these extra conservation laws arise from new relabelling
symmetries that exist in the slice model framework.

\section{Relabelling symmetries and Noether's theorem for slice models}
\label{sec:relabelling}
In this section we describe the relabelling symmetry for slice models
and compute the corresponding conserved quantities via Noether's
theorem. Relabelling symmetry in fluid dynamics is the statement that
the reference configuration for the Lagrangian flow map is
arbitrary. In the slice framework, this means that we can arbitrarily
select an alternative reference configuration, which can be
transformed back to the original reference configuration by the slice
map represented by the relabelling transformation $(\psi,g)\in
\Diff(\Omega)\circledS \mathcal{F}(\Omega)$. After this change
of base coordinates, the Lagrangian flow map is transformed according
to
\begin{equation}
  (\phi, f) \mapsto (\phi,f)(\psi,g) = (\phi\circ\psi, f \circ \psi + g).
\end{equation}
\revised{Relabelling symmetries are such transformations that leave the initial data $\theta_{S,0},D_0,s_0$
  all invariant.
  In the group variable notation of previous sections, relabelling symmetries form a group
  \begin{equation}
    G_{D_0,\theta_{S,0},s_0} =
    \left\{
    (\psi,g) \in \Diff(\Omega)\circledS \mathcal{F}(\Omega) \quad | \quad \psi^*(D_0\diff S) = D_0\diff S
    \mbox{ and } g = (\theta_{S,0}-\psi^*\theta_{S,0}/s_0\right\}.
  \end{equation}
  }

To compute infinitesimal relabelling transformations, we consider a
1-parameter family of relabellings $(\psi_\epsilon,g_\epsilon)$ for
$\epsilon>0$, with $(\psi_0,g_0)=(\Id,0)$, so that
\begin{equation}
  (\psi_\epsilon, g_\epsilon) = (\Id,0) + \epsilon(v_S,v_T) +
  \mathcal{O}(\epsilon).
\end{equation}
Then,
\begin{align}
 \nonumber \delta(\phi, f) & = \lim_{\epsilon\to 0}\frac{1}{\epsilon}\left(
  (\phi\circ\psi_\epsilon, f\circ \psi_\epsilon + g_\epsilon)-(\phi,f)\right)\\
  & = \left(\nabla\phi\cdot v_S, \nabla f\cdot v_S + v_T\right).
\end{align}
Following \eqref{eq:flow map}, we define $(w_S,w_T)$ to be the unique element of
$\mathfrak{X}(\Omega)\circledS \mathcal{F}(\Omega)$ such that
\begin{equation}
  w_S\circ \phi = \nabla\phi\cdot v_S := \delta\phi, \quad
  w_T\circ \phi = \nabla f\cdot v_S + v_T := \delta f.
\end{equation}
Time differentiation gives
\begin{align}
\nonumber  \pp{}{t}(w_S\circ\phi)
  &= \pp{w_S}{t}\circ\phi + u_S\circ\phi\cdot(\nabla w_S)\circ\phi, \\
  \revised{(}\pp{}{t} w_T\circ \phi) &= 
  \pp{w_T}{t}\circ\phi + u_S\circ\phi\cdot(\nabla w_T)\circ\phi.
\end{align}
On the other hand,
\begin{align}
\nonumber  \pp{}{t}(\nabla\phi\cdot v_S) &=
  \nabla(u_S\circ \phi)\cdot v_S, \\
  & = (\nabla u_S)\circ \phi \cdot (\nabla\phi\cdot v_S), \\
\nonumber  & = (\nabla u_S)\circ \phi \cdot w_S\circ\phi, \\
\nonumber  \pp{}{t}(w_T \circ \phi) & =
  \nabla (u_T\circ \phi)\cdot v_S, \\
\nonumber  & = (\nabla u_T)\circ \phi \cdot (\nabla\phi\cdot v_S), \\
  & = (\nabla u_T)\circ \phi \cdot w_S\circ\phi.
\end{align}
Equating and composing with $\phi^{-1}$, we obtain
\begin{equation}
  \pp{}{t}(w_S,w_T) + [(u_S,u_T),(w_S,w_T)] = 0.
\end{equation}
This means that 
\begin{align}
  \label{eq:relabelling S}
  \delta u_S & = \dot{w}_S + [u_S,w_S] = 0, \\
  \label{eq:relabelling T}
  \delta u_T & = \dot{w}_T + u_S\cdot\nabla w_T - w_S\cdot\nabla u_T = 0,
\end{align}
\emph{i.e.} $u_S$ and $u_T$ are left invariant under relabelling
transformations.

To leave $\theta_S$ invariant, the infinitesimal 
symmetries $(w_S,w_T)$ also need to satisfy
\begin{equation}
  \label{eq:relabelling theta}
  \delta \theta_S = -L_{(w_S,w_T)}(\theta_S,s) = -w_S\cdot\nabla \theta_S -sw_T = 0.
\end{equation}
In the $s=0$ case, this restricts $w_S$ to being tangential to the contours
of $\theta$, hence the baroclinic torque term. However, if $s\neq 0$,
we can take \emph{any} $w_S$ and then pick
\begin{equation}
  \label{eq:w_T theta}
  w_T = -\frac{1}{s}w_S\cdot\nabla \theta_S,
\end{equation}
i.e. any change in $\theta_S$ caused by advection with $w_S$ can be
corrected by a source term from $w_T$. For this to work, we need
equation \eqref{eq:w_T theta} to be compatible with
\eqref{eq:relabelling T}. This is verified by the following
proposition.
\begin{proposition}
  \label{eq:things commute}
  Let $w_S$ be a vector field with arbitrary initial condition, and
  let $w_T$ be a function satisfying \eqref{eq:w_T theta}
  initially. Let $(w_S,w_T)$ satisfy the particle relabelling symmetry
  conditions (\ref{eq:relabelling S}-\ref{eq:relabelling T}), and let
  $\theta$ evolve according to
  \begin{equation}
    \label{eq:dtheta/dt}
    \pp{\theta_S}{t} = -L_{(u_S,u_T)}(\theta_S,s) = -(u_S \cdot \nabla \theta_S + su_T).
  \end{equation}
  Then $\theta_S$, $w_T$ satisfy \eqref{eq:w_T theta} for all time.
\end{proposition}
\begin{proof}
  First note that $\mathcal{L}_{(u_S,u_T)}(\theta_S,s)$ is a Lie algebra action
  of $\mathfrak{X}(\Omega)\circledS \mathcal{F}(\Omega)$ on
  $\mathcal{F}(\Omega)\times \mathbb{R}$, i.e.
  \begin{equation}
    \mathcal{L}_{(u_S,u_T)}\mathcal{L}_{(w_S,w_T)}(\theta_S,s) -
    \mathcal{L}_{(w_S,w_T)}\mathcal{L}_{(u_S,u_T)}(\theta_S,s)
    = \mathcal{L}_{[(u_S,u_T),(w_S,w_T)]},
  \end{equation}
  where $[(u_S,u_T),(w_S,w_T)]$ is the bracket on $\mathfrak{X}(\Omega)\circledS \mathcal{F}(\Omega)$ defined in \cite{cotter2013variational}, given by
  \begin{equation*}
[({u}_S,u_T),({w}_S,w_T)] = \left([{u}_S,{w}_S] ,\,{u}_S\cdot\nabla w_T - {w}_S\cdot\nabla u_T \right).
\label{SDalgebraaction}
  \end{equation*}
  Then we have
  \begin{align}
\nonumber    \pp{}{t}\mathcal{L}_{(w_S,w_T)}(\theta_S,s) &
    = \mathcal{L}_{\pp{}{t}(w_S,w_T)}(\theta_S,s) +
    \mathcal{L}_{(w_S,w_T)}(\dot{\theta}_S,0), \\
\nonumber    &= \mathcal{L}_{\pp{}{t}(w_S,w_T)}(\theta_S,s) +
    \mathcal{L}_{(w_S,w_T)}\mathcal{L}_{(u_S,u_T)}(\theta_S,s), \\
\nonumber    &= \mathcal{L}_{\pp{}{t}(w_S,w_T)}(\theta_S,s) -
    \mathcal{L}_{(w_S,w_T)}\mathcal{L}_{(u_S,u_T)}(\theta_S,s), \\
 \nonumber   &= \mathcal{L}_{\pp{}{t}(w_S,w_T)}(\theta_S,s) -
    \mathcal{L}_{[(u_S,u_T),(w_S,w_T)]}(\theta_S,s) -
    \mathcal{L}_{(u_S,u_T)}\underbrace{\mathcal{L}_{(w_S,w_T)}(\theta_S,s)}_{=0}, \\
    &= \mathcal{L}_{\underbrace{\pp{}{t}(w_S,w_T)+[(u_S,u_T),(w_S,w_T)]}_{=0}}(\theta_S,s)=0,
  \end{align}
  as required.
\end{proof}
\begin{corollary}
  Let $u_S,u_T,\theta_S,D$ solve the equations (\ref{eq:EPS}-\ref{eq:EPT}). Then the potential vorticity
\revised{$q$}  (weakly) satisfies the Lagrangian conservation law 
  \begin{equation}
    \pp{q}{t} + u_S\cdot\nabla q = 0,
  \end{equation}
  \revised{as a consequence of Noether's theorem.}
\end{corollary}
\begin{proof}
  \revised{
  Given initial condition $D_0$ for density, we pick arbitrary $\psi_0 \in \Lambda^0(\Omega)$, compactly supported in the interior of $\Omega$.
  Then we choose $\psi$ as the solution of
  \begin{equation}
    \label{eq:psi transport}
    \left(\pp{}{t} + \mathcal{L}_{u_S}\right)\psi = 0,
  \end{equation}
  with initial condition $\psi_0$. 
  We then define $w_s$ via
  \begin{equation}
    \label{eq:d psi}
    w_s \lrcorner D \diff S = \diff \psi.
  \end{equation}
  Then
  \begin{align}
    \left(\pp{}{t} + \mathcal{L}_{u_S}\right)w_s \lrcorner D \diff S
    &=\left(\pp{}{t} + \mathcal{L}_{u_S}\right)\diff \psi, \\
    &=\diff\left(\pp{}{t} + \mathcal{L}_{u_S}\right)\psi = 0,
  \end{align}
  and we deduce from the chain rule and Equation \eqref{eq:D} that
  $w_s$ satisfies Equation \eqref{eq:relabelling S}. Further,
  Equation \eqref{eq:d psi} implies that 
  \begin{equation}
    \label{eq:persistence}
    \delta(D \diff S) = \mathcal{L}_{w_S}(D\diff S) = d(w_S\lrcorner(D\diff S)) = 0,
  \end{equation}
  for all times, \emph{i.e.} $w_s$ is a relabelling symmetry for $D$.
  We then choose
  \begin{equation}
    w_T = -\frac{1}{s}w_S\cdot\nabla\theta_S,
  \end{equation}
  which defines a relabelling symmetry for $\theta_S$ by Lemma \ref{eq:things commute}.}

  Next, we follow the steps of Noether's Theorem, considering the
  variations in the action $S$ under the relabelling transformations
  generated by $(w_S,w_T)$ defined above. Since these transformations
  leave $(u_S,u_T)$, $D$ and $\theta_S$ invariant, the action does not
  change, and we get
  \begin{align}
 \nonumber   0  &= \delta S \\ 
 \nonumber   &= \int_{t_0}^{t_1} \int_\Omega \dede{l}{u_S}\cdot \delta u_S + \dede{l}{u_T}\delta u_T
    + \dede{l}{D}\delta D + \dede{l}{\theta_S}\delta\theta_S + \dede{l}{s}\delta s
    \diff S \diff t, \\
\nonumber    &= \int_{t_0}^{t_1} \int_\Omega \dede{l}{u_S}\cdot \left(\pp{}{t}w_S + [u_S,w_S]\right)\diff S + \dede{l}{u_T}\left(\pp{}{t}w_T + \mathcal{L}_{u_S}w_T 
    - \mathcal{L}_{w_S}u_T\right)\diff S \\
  \nonumber  & \qquad  - \dede{l}{D}\mathcal{L}_{w_S}(D\diff S) -
    \mathcal{L}_{(w_S,w_T)}\theta_S
    \dede{l}{\theta_S}\diff S \diff t, \\
 \nonumber   &= \left[\int_\Omega
      \dede{l}{u_S}\cdot w_S + \dede{l}{u_T}w_T \diff S
      \right]_{t=t_1}^{t=t_2} \nonumber \\
 \nonumber   &\qquad -\int_{t_0}^{t_1} \int_\Omega w_S\cdot \left(
    \left(\pp{}{t} + \mathcal{L}_{u_S}\right)\dede{l}{u_S}\cdot\diff x \otimes
    \diff S + \dede{l}{u_T}\diff u_T\otimes \diff S + \dede{l}{\theta_S}\diff \theta_S \otimes \diff S
    - D\diff \dede{l}{D} \otimes \diff S\right) \nonumber \\
 \nonumber   &\qquad  + \int_\Omega w_T\left(\left(\pp{}{t} + \mathcal{L}_{u_S}\right)\dede{l}{u_T}\diff S 
    + \dede{l}{\theta_S}s\diff S\right)\diff t, \\
    &= \left[\int_\Omega
      \dede{l}{u_S}\cdot w_S + \dede{l}{u_T}w_T
      \right]_{t=t_1}^{t=t_2} \diff S,
  \end{align}
  since $(u_S,u_T)$, $D$ and $\theta_S$ solve (\ref{eq:EPS}-\ref{eq:EPT}).
  For sufficiently smooth solutions in time, we may consider the limit $t_1\to t_2$, and we get
  \begin{align}
\nonumber    0 & = \dd{}{t}\int_\Omega \dede{l}{u_S}\cdot w_S + \dede{l}{u_T}\cdot w_T\diff
    V, \\
\nonumber     & = \dd{}{t}\int_\Omega \left(\dede{l}{u_S}\cdot\diff x - \revised{\frac{1}{s}\dede{l}{u_T}}\diff \theta
    \right)\wedge w_S\lrcorner \diff S, \\
 \nonumber   & = \dd{}{t}\int_\Omega \frac{1}{D}\left(\dede{l}{u_S}\cdot\diff x - \revised{\frac{1}{s}\dede{l}{u_T}}\diff \theta
    \right)\wedge w_S\lrcorner D \diff S, \\
  \nonumber  & = \dd{}{t}\int_\Omega \frac{1}{D}\left(\dede{l}{u_S}\cdot\diff x - \revised{\frac{1}{s}\dede{l}{u_T}}\diff \theta
    \right)\wedge \diff \psi, \label{eq:psi KCT} \\
    \nonumber & = \int_\Omega \pp{}{t}\frac{1}{D}\left(\dede{l}{u_S}\cdot\diff x - \revised{\frac{1}{s}\dede{l}{u_T}}\diff \theta
    \right)\wedge \diff \psi
    + \frac{1}{D}\left(\dede{l}{u_S}\cdot\diff x - \revised{\frac{1}{s}\dede{l}{u_T}}\diff \theta
    \right)\wedge \pp{}{t} \diff \psi,\\
    \nonumber & = \int_\Omega \pp{}{t}\frac{1}{D}\left(\dede{l}{u_S}\cdot\diff x - \revised{\frac{1}{s}\dede{l}{u_T}}\diff \theta
    \right)\wedge \diff \psi
    - \frac{1}{D}\left(\dede{l}{u_S}\cdot\diff x - \revised{\frac{1}{s}\dede{l}{u_T}}\diff \theta
    \right)\wedge \mathcal{L}_{u_S}\diff \psi,\\
    & = -\int_\Omega \psi\left(\pp{}{t} + \mathcal{L}_{u_S}\right)\diff\frac{1}{D}\left(\dede{l}{u_S}\cdot\diff x - \revised{\frac{1}{s}\dede{l}{u_T}}\diff \theta
    \right),
  \end{align}
  \revised{
  which holds for arbitrary $\psi$, hence we deduce that
  \begin{align}
\nonumber    0 & = \left(\pp{}{t} + \mathcal{L}_{u_S}\right)(qD \diff S), \\
    & = \left(\pp{q}{t} + \mathcal{L}_{u_S}q\right)(D \diff S),
  \end{align}
  hence the result (since $D$ is positive).}
\end{proof}

\section{Summary and outlook}
\label{sec:summary}
  In this paper, we provided a new proof that the conserved potential
  vorticity in vertical slice models arises through Noether's Theorem
  upon consideration of the relabelling symmetries consisting of
  rearrangements of the vertical slice combined with transverse motion
  that restores the original structure of the potential temperature
  $\theta_S$. The proof applies to a horizontally-periodic geometry
  with rigid top and bottom boundary but can be easily extended to
  other slice geometries.

  A future direction will be to use this variational structure to
  build potential vorticity conserving vertical slice models, and
  to make further comparisons with the optimal transport formulation.
  It would also be interesting to use the relabelling transformations
  in this paper to derive balanced models using the variational asymptotics
  approach of \cite{oliver2006variational}.

\bibliography{slice_variational}

\begin{thebibliography}{HMR98}

\bibitem[CH13a]{cotter2013noether}
C.~J. Cotter and D.~D. Holm.
\newblock On {N}oether’s theorem for the {Euler--Poincar{\'e}} equation on
  the diffeomorphism group with advected quantities.
\newblock {\em Foundations of Computational Mathematics}, 13(4):457--477, 2013.

\bibitem[CH13b]{cotter2013variational}
C.~J. Cotter and D.~D. Holm.
\newblock A variational formulation of vertical slice models.
\newblock {\em Proc. R. Soc. A}, 469(2155):20120678, 2013.

\bibitem[Cul06]{Cu2006}
M.~J.~P. Cullen.
\newblock {\em A Mathematical Theory of Large-Scale Atmospheric Flow}.
\newblock Imperial College Press, 2006.

\bibitem[Cul07]{Cu2007}
M.~J.~P. Cullen.
\newblock Modelling atmospheric flows.
\newblock {\em Acta Numerica}, 2007.

\bibitem[Cul08]{Cu2008}
M.~J.~P. Cullen.
\newblock A comparison of numerical solutions to the {E}ady frontogenesis
  problem.
\newblock {\em Q. J. R. Meteorol. Soc.}, 134:2143--2155, 2008.

\bibitem[HMR98]{HoMaRa1998}
D.~D. Holm, J.~E. Marsden, and T.~S. Ratiu.
\newblock The {Euler--Poincar\'e} equations and semidirect products with
  applications to continuum theories.
\newblock {\em Adv. in Math.}, 137:1--81, 1998.

\bibitem[Oli06]{oliver2006variational}
M.~Oliver.
\newblock Variational asymptotics for rotating shallow water near geostrophy: a
  transformational approach.
\newblock {\em Journal of Fluid Mechanics}, 551:197--234, 2006.

\bibitem[VCC14]{visram2014framework}
A.~R. Visram, C.~J. Cotter, and M.~J.~P. Cullen.
\newblock A framework for evaluating model error using asymptotic convergence
  in the eady model.
\newblock {\em Quarterly Journal of the Royal Meteorological Society},
  140(682):1629--1639, 2014.

\end{thebibliography}

\end{document}